\pdfoutput=1
\RequirePackage{ifpdf}
\ifpdf 
\documentclass[pdftex]{sigma}
\else
\documentclass{sigma}
\fi

\numberwithin{equation}{section}

\newtheorem{Theorem}{Theorem}[section]
\newtheorem{Corollary}[Theorem]{Corollary}
\newtheorem{Lemma}[Theorem]{Lemma}
\newtheorem{Proposition}[Theorem]{Proposition}
 { \theoremstyle{definition}
\newtheorem{Note}[Theorem]{Note}
\newtheorem{Example}[Theorem]{Example} }

\newcommand{\pFq}[5]{\ensuremath{{}_{#1}F_{#2} \left( \genfrac{}{}{0pt}{}{#3}
{#4} \bigg| {#5} \right)}}

\begin{document}

\allowdisplaybreaks

\newcommand{\arXivNumber}{1505.06274}

\renewcommand{\PaperNumber}{001}

\FirstPageHeading

\ShortArticleName{The Moments of the Hydrogen Atom by the Method of Brackets}

\ArticleName{The Moments of the Hydrogen Atom\\ by the Method of Brackets}

\Author{Ivan GONZALEZ~$^{\dag^1}$, Karen T.~KOHL~$^{\dag^2}$, Igor KONDRASHUK~$^{\dag^3}$, Victor H.~MOLL~$^{\dag^4}$\\ and Daniel SALINAS~$^{\dag^5}$}

\AuthorNameForHeading{I.~Gonzalez, K.T.~Kohl, I.~Kondrashuk, V.H.~Moll and D.~Salinas}

\Address{$^{\dag^1}$~Instituto de F\'{i}sica y Astronomia, Universidad de Valparaiso,\\
\hphantom{$^{\dag^1}$}~Avda.\ Gran Breta\~{n}a 1111, Valparaiso, Chile}
\EmailDD{\href{mailto:ivan.gonzalez@uv.cl}{ivan.gonzalez@uv.cl}}

\Address{$^{\dag^2}$~Department of Mathematics, University of Southern Mississippi,\\
\hphantom{$^{\dag^2}$}~Long Beach, MS 39560, USA}
\EmailDD{\href{mailto:karen.kohl@usm.edu}{karen.kohl@usm.edu}}

\Address{$^{\dag^3}$~Grupo de Matem\'atica Aplicada {\rm \&} Grupo de F\'isica de Altas Energ\'ias, \\
\hphantom{$^{\dag^3}$}~Departmento de Ciencias B\'{a}sicas, Universidad del B\'{i}o-B\'{i}o, Campus Fernando May,\\
\hphantom{$^{\dag^3}$}~Av.~Andres Bello 720, Casilla 447, Chill\'{a}n, Chile}
\EmailDD{\href{mailto:igor.kondrashuk@gmail.com}{igor.kondrashuk@gmail.com}}

\Address{$^{\dag^4}$~Department of Mathematics, Tulane University, New Orleans, LA 70118, USA}
\EmailDD{\href{mailto:vhm@tulane.edu}{vhm@tulane.edu}}
\URLaddressDD{\url{http://129.81.170.14/~vhm/}}

\Address{$^{\dag^5}$~Departamento de Fisica, Universidad T\'{e}cnica {F}ederico {S}anta {M}ar\'{i}a,\\
\hphantom{$^{\dag^5}$}~Casilla 110-V, Valparaiso, Chile}
\EmailDD{\href{mailto:salinas.a.daniel@gmail.com}{salinas.a.daniel@gmail.com}}

\ArticleDates{Received November 23, 2016, in f\/inal form December 30, 2016; Published online January 05, 2017}

\Abstract{Expectation values of powers of the radial coordinate in arbitrary hydrogen states are given, in the quantum case, by an integral
involving the associated Laguerre function. The method of brackets is used to evaluate the integral in
closed-form and to produce an expression for this average value as a f\/inite sum.}


\Classification{33C45; 33C20; 81V45}

\section{Introduction}\label{sec-intro}

The computation of the expectation $\langle r^{k} \rangle$ of the electron for atoms with a~single electron is a~standard problem in quantum mechanics, see \cite{pasternak-1937a,vanvleck-1934a}. For a given energy state $n$, the problem is expressed as
\begin{gather*}
\big\langle r^{k} \big\rangle = \int_{0}^{\infty} R_{n \ell}^{2}(r) r^{k+2} dr,
\end{gather*}
where $R_{n \ell}(r)$ is the radial solution of the Schr\"{o}dinger equation for the hydrogen atom. Conditions on the parameters $n$, $\ell$, $k$ are determined by the convergence of this integral.

In the non-relativistic situation, the solution is given in terms of the Hahn polynomials \cite{andrews-1999a}:
\begin{gather}
h_{m}^{(\alpha, \beta)}(x,N) = \frac{(1-N)_{m} (\beta+1)_{m}}{m!} \, \pFq32{-m, \,\alpha + \beta + m + 1,\, -x}{\beta + 1,\, 1 - N}{\,1}.
\label{hahn-1}
\end{gather}
In particular, these expectations are given in terms of the Chebyshev polynomials of discrete variables \cite{nikiforov-1991a, suslov-2010a}
 \begin{gather*}
 t_{m}(x,N) = h_{m}^{(0,0)}(x,N)
 \end{gather*}
 in the form
 \begin{gather}\label{form-f32a1}
\big\langle r^{k} \big\rangle = \frac{1}{2n} (2 \mu)^{-k} t_{k+1}(n- \ell -1, - 2 \ell -1), \qquad \text{when} \quad k = -1, 0, 1, 2, \dots,
\end{gather}
and
\begin{gather}\label{form-f32b}
\big\langle r^{k} \big\rangle = \frac{1}{2n} (2 \mu)^{-k} t_{-k-2}(n- \ell -1, - 2 \ell -1), \qquad \text{when} \quad k = -2, -3, \dots, -2 \ell-2.
\end{gather}
The parameters are $\mu = Z/na_{0}$ with $a_{0} = \hbar^{2}/me^{2}$ the Bohr radius and $Z$ is the nuclear charge. The constants $m$ and $e$ are the mass and
charge of the electron, respectively.

The identity
\begin{gather*}
t_{k}(n - \ell -1, -2 \ell -1) =\frac{\Gamma(2 \ell + k + 2)}{\Gamma(2 \ell + 2)}\,
\pFq32{-k, \, k+1, \, -n+ \ell + 1 }{1, \, 2 \ell + 2}{\,1}
\end{gather*}
follows from \eqref{hahn-1}. Then \eqref{form-f32a1} becomes
\begin{gather*}
\big\langle r^{k}\big \rangle_{n \ell} = \frac{1}{2n (2 \mu)^{k}} \frac{\Gamma(2 \ell+k+3)}{\Gamma(2 \ell+2)}\,
\pFq32{-1-k, \, k+2, \, -n+\ell+1}{1, \,2 \ell+2}{\,1}
\end{gather*}
for $k = -1, 0, 1, 2, \dots$, and \eqref{form-f32b}
\begin{gather*}
\big\langle r^{k} \big\rangle_{n \ell} = \frac{1}{2n(2 \mu)^{k}}\frac{\Gamma(2 \ell -k )}{\Gamma(2 \ell + 2)}\,
\pFq32{-2+k, \, -k+3, \, -n+\ell+1}{1, \, 2 \ell+2}{\,1}
\end{gather*}
for $k = -2, -3, \dots, -2 \ell - 2$, where the dependence upon the parameters $n$ and $\ell$ have been made explicit.

In the quantum case, the radial component of the wave function for a hydrogen atom with nuclear charge $Z$ is characterized by two quantum numbers: $n$ the
\textit{principal quantum number} and $\ell$ the \textit{orbital number}. The corresponding normalized radial function is
\begin{gather*}
R_{n \ell}(r) = A_{n \ell} ( 2 \mu r)^{\ell} \exp (- \mu r) L_{n- \ell - 1}^{2 \ell+1}( 2 \mu r),
\end{gather*}
where the normalization constant is
\begin{gather*}
A_{n \ell} = \sqrt{ \frac{(2 \mu)^{3}}{2n} \frac{(n- \ell -1)!}{(n+ \ell)!}}
\end{gather*}
and
\begin{gather*}
L_{m}^{\alpha}(x) = \frac{\Gamma( \alpha + m +1)}{\Gamma(m+1) \Gamma(1+ \alpha)} \, \pFq11{-m}{1+\alpha}{\,x}
\end{gather*}
is the associated Laguerre function; see \cite[formula 8.972.1]{gradshteyn-2015a}. The expectation value of a power of the radial coordinate is given by
\begin{gather}
\big \langle r^{k} \big \rangle_{n \ell} = ( 2 \mu)^{2 \ell} A_{n \ell}^{2} \int_{0}^{\infty} r^{2 + 2 \ell +k}
e^{-2 \mu r} \big[ L_{n- \ell -1}^{2 \ell +1}(2 \mu r) \big]^{2} dr,
\label{mean-1}
\end{gather}
with $n \in \mathbb{N}$, $0 \leq \ell \leq n-1$, and $k \in \mathbb{Z}$.

The goal of the work is to compute the integral in \eqref{mean-1} by the \textit{method of brackets}, to illustrate its f\/lexibility. The reader will f\/ind in \cite{amdeberhan-2012b,gonzalez-2014a,gonzalez-2010a,gonzalez-2010b,gonzalez-2007a,gonzalez-2008a,kohl-2011b} a~collection of examples of def\/inite integrals evaluated by this method. The basic procedure is described in Section~\ref{sec-method}.

The examples presented here are to be considered as the beginning of a series of calculations of integrals related to the hydrogen atom. These include the evaluation of the integral~\cite{gonzalez-2009a}
\begin{gather*}
J_{nms}^{\alpha \beta} = \int_{0}^{\infty} e^{-x}x^{\alpha +s} L_{n}^{\alpha}(x) L_{m}^{\beta}(x) dx
\end{gather*}
given by S.K.~Suslov and B.~Trey \cite{suslov-2008a}. The method of brackets provides an alternative method of proof that \textit{only uses} the hypergeometric representation of the Laguerre function. The method can also be used to discuss the relativistic situation. Details will appear elsewhere.

The reductions of the formulas discussed here uses basic properties of the gamma function, such as
\begin{gather}\label{gamma-1}
\Gamma(a+n) = \Gamma(a) (a)_{n} \qquad \text{and}\qquad (a)_{-n} = \frac{(-1)^{n}}{(1-a)_{n}}\qquad \text{for} \quad a \in \mathbb{R},
\quad n \in \mathbb{N}.
\end{gather}
Here $(a)_{n} = a(a+1) \cdots (a+n-1)$ is the Pochhammer symbol.

\section{A direct evaluation}\label{sec-direct}

This section presents a direct evaluation of the integral
\begin{gather}\label{mean-2}
\big \langle r^{k} \big \rangle_{n \ell} = ( 2 \mu)^{2 \ell} A_{n \ell}^{2} \int_{0}^{\infty} r^{2 + 2 \ell +k}
e^{-2 \mu r} \big[ L_{n- \ell -1}^{2 \ell +1}(2 \mu r) \big]^{2} dr
\end{gather}
given in \eqref{mean-1}. The proof is based on some identities for the associated Laguerre function appearing in the integrand. The methods presented here are then compared with the evaluation by the \textit{method of brackets} explained in the next section.

The f\/irst identity used to modify the integrand appears in \cite[formula 8.976.3]{gradshteyn-2015a}
\begin{gather}\label{laguerre-1}
\big[ L_{m}^{\alpha}(x) \big]^{2} =\frac{\Gamma(\alpha + m + 1)}{2^{2m} \Gamma(m+1)}\sum_{s=0}^{m} \binom{2m-2s}{m-s}
\frac{\Gamma(2s+1)}{\Gamma(\alpha + s + 1) \Gamma(s+1)}L_{2s}^{2 \alpha}(2x).
\end{gather}

Therefore
\begin{gather}
\big \langle r^{k} \big \rangle_{n \ell} = ( 2 \mu)^{2 \ell} A_{n \ell}^{2} \frac{\Gamma(\ell + n + 1 )}{2^{2(n-\ell-1)} \Gamma(n - \ell)} \nonumber\\
\hphantom{\big \langle r^{k} \big \rangle_{n \ell} =}{} \times \sum_{s=0}^{n - \ell -1}
\binom{2(n-\ell - 1 - s)}{n-\ell - 1 - s} \frac{\Gamma(2s+1)}{\Gamma(2 \ell + 2 + s) \Gamma(s+1)}G_{\ell,k,s}(\mu),\label{rwithG}
\end{gather}
where
\begin{gather}\label{def-G}
G_{\ell,k,s}(\mu) = \int_{0}^{\infty} r^{2 + 2 \ell + k} e^{-2 \mu r}L_{2s}^{2( 2 \ell + 1)}(4 \mu r ) dr.
\end{gather}

To obtain an expression for $G_{\ell,k,s}(\mu)$, the representation
\begin{gather}\label{hyper-laguerre}
L_{n}^{a}(x) = \frac{\Gamma(a+n+1)}{\Gamma(n+1) \Gamma(1+a)} \, \pFq11{-n}{1+a}{\,x}
\end{gather}
for the Laguerre function (see \cite[formula 8.972.1]{gradshteyn-2015a}) is used.

\begin{Theorem}
The integral $G_{\ell,k,s}(\mu)$ is given by
\begin{gather*}
G_{\ell,k,s}(\mu) =\frac{\Gamma(4 \ell + 2s + 3) \Gamma(2 \ell + k + 3)}{\Gamma(2s+1)
\Gamma(4 \ell + 3) (2 \mu)^{2 \ell + k + 3}}\,
\pFq21{{-2s},\, {2 \ell + k + 3}}{4 \ell + 3}{ \, 2}.
\end{gather*}
\end{Theorem}
\begin{proof}
The hypergeometric representation \eqref{hyper-laguerre} shows that
\begin{gather*}
L_{2s}^{2(2 \ell+1)}(4 \mu r) =
\frac{\Gamma(4 \ell + 2s + 3)}{\Gamma(2s+1) \Gamma(4 \ell + 3 )}\,
\pFq11{-2s}{4 \ell + 3 }{\,4 \mu r}.
\end{gather*}
Expanding the hypergeometric function gives
\begin{align*}
G_{\ell, k ,s}(\mu) & = \frac{\Gamma(4 \ell + 2s + 3)}{\Gamma(2s+1)\Gamma(4 \ell + 3)}\int_{0}^{\infty} \sum_{j=0}^{2s}
\frac{(-2s)_{j}}{(4 \ell + 3 )_{j}} \frac{ (4 \mu r)^{j}}{j!}r^{2 \ell + 2 + k} e^{-2 \mu r} dr \\
& = \frac{\Gamma(4 \ell + 2s + 3)}{\Gamma(2s+1)\Gamma(4 \ell + 3)} \sum_{j=0}^{2s}\frac{(-2s)_{j}}{(4 \ell + 3 )_{j} j!}(4 \mu)^{j}\int_{0}^{\infty}
r^{2 \ell + 2 + k+j} e^{-2 \mu r} dr \\
& = \frac{\Gamma(4 \ell + 2s + 3)}{\Gamma(2s+1)\Gamma(4 \ell + 3)} \sum_{j=0}^{2s}\frac{(-2s)_{j} (4 \mu)^{j}}{(4 \ell + 3 )_{j} j!}
\frac{\Gamma(2 \ell + k +j + 3)}{(2 \mu)^{2 \ell + k + j + 3}} \\
& = \frac{\Gamma(4 \ell + 2s + 3)}{\Gamma(2s+1)\Gamma(4 \ell + 3) (2 \mu)^{2 \ell + k + 3}} \sum_{j=0}^{2s}\frac{(-2s)_{j} 2^{j}}{(4 \ell + 3 )_{j} j!}
\Gamma(2 \ell + k +j + 3) \\
& = \frac{\Gamma(4 \ell + 2s + 3) \Gamma(2 \ell + k + 3)}{\Gamma(2s+1)\Gamma(4 \ell + 3) (2 \mu)^{2 \ell + k + 3}} \sum_{j=0}^{2s}
\frac{(-2s)_{j} (2 \ell + k + 3)_{j} }{(4 \ell + 3 )_{j} j!} 2^{j} \\
& = \frac{\Gamma(4 \ell + 2s + 3) \Gamma(2 \ell + k + 3)}{\Gamma(2s+1)\Gamma(4 \ell + 3) (2 \mu)^{2 \ell + k + 3}} \sum_{j=0}^{\infty}
\frac{(-2s)_{j} (2 \ell + k + 3)_{j} }{(4 \ell + 3 )_{j} j!} 2^{j} \\
& = \frac{\Gamma(4 \ell + 2s + 3) \Gamma(2 \ell + k + 3)}{\Gamma(2s+1)\Gamma(4 \ell + 3) (2 \mu)^{2 \ell + k + 3}}\,
\pFq21{{-2s},\, {2 \ell + k + 3}}{4 \ell + 3}{\, 2}.
\end{align*}
This is the stated form for $G_{\ell,k,s}(\mu)$.
\end{proof}

\begin{Note}
Observe that $s \in \mathbb{N}$, so the hypergeometric function in the expression for $G_{\ell,k,s}(\mu)$ is actually a polynomial in its last variable. Thus, there are no convergence issues.
\end{Note}

The expression for $G_{\ell,k,s}(\mu)$ and~\eqref{rwithG} are used to produce the next result (after the change $s \mapsto n - \ell - 1 -s$).

\begin{Corollary}\label{moments-2}
For $n = 1, 2, \dots$, $\ell = 0, 1, \dots, n-1$ and $k \in \mathbb{Z}$ with $2 \ell + k + 3 > 0$. The moments of the hydrogen atom are given by
\begin{gather*}
\big\langle r^{k} \big\rangle_{n \ell} =\frac{\Gamma(2 \ell + k + 3) (2n+2 \ell)! } {n 2^{2n-2\ell-1} (4 \ell +2)! (2 \mu)^{k} (n+ \ell)! (n-\ell -1)!} \\
\hphantom{\big\langle r^{k} \big\rangle_{n \ell} =}{}\times
\sum_{s=0}^{n- \ell-1} \frac{\binom{n + \ell}{s} \binom{n-\ell-1}{s}}{\binom{2n+2 \ell}{2s}} \, \pFq21{{-2(n-\ell-1-s)},\, {2 \ell + k +3}}{4 \ell + 3}{\,2}.
\end{gather*}
\end{Corollary}

\begin{Note}
The restriction $2 \ell +k + 3>0$ avoids the singularities of the gamma factor $\Gamma(2 \ell +k + 3)$. Also observe that the f\/irst entry in the series ${_{2}F_{1}}$ in the answer is a negative integer, therefore the series reduces to a f\/inite sum.
\end{Note}

In this article the expectation values of the powers of the radial coordinate of the hydrogen atom in a framework of quantum mechanics, that is, in the non-relativistic case are computed. In the Introduction it was stated that this already has appeared in the literature. In the relativistic case, results for these expectation values of the powers of the radial coordinate appeared in 2009. Indeed, the relativistic Coulomb integrals are contained in \cite{suslov-2009a, suslov-2010b}. The treatment of the results obtained in \cite{suslov-2009a, suslov-2010b} by computer algebra methods is described in \cite{suslov-2014a, suslov-2013a}.

In the nonrelativistic case of quantum mechanics, the corresponding questions were successfully solved by direct calculation. For example, in~\cite{dehesa-1997a} useful relations between dif\/ferent Laguerre polynomials were found. In \cite{dehesa-2010a} the radial expectation values are given for $D$-dimensional hydrogenic states with $D > 1$. The same quantities are discussed in a more general setting in \cite{dehesa-2016b, dehesa-2016a}. The radial expectation values of hydrogenic states in momentum space appear in~\cite{dehesa-2000a}, represented in terms of Gegenbauer polynomials instead of Laguerre polynomials. All these results were obtained by direct calculations too. The method of brackets may signif\/icantly simplify the calculations for these tasks. This will be discussed in a future publication.

The method of brackets is not the unique successful method which involves integral transformations. Traditional methods based on Mellin--Barnes transformation may be ef\/f\/icient tools in order to obtain new results in quantum f\/ield theory \cite{Allendes:2009bd,Allendes:2012mr,Gonzalez:2012gu,Gonzalez:2012wk,Kniehl:2013dma}.

\section{The method of brackets}\label{sec-method}

The evaluation of the integral giving the mean value $\langle r^{k} \rangle $ \eqref{mean-2} presented in the previous section, used the relation \eqref{laguerre-1} in a~fundamental way. A method to evaluate integrals over the half line $[0, \infty)$, based on a small number of rules has been developed in \cite{gonzalez-2007a,gonzalez-2008a}. This \textit{method of brackets} is described next. The heuristic rules are currently being placed on solid ground \cite{amdeberhan-2012b}. The reader will f\/ind in \cite{gonzalez-2014a, gonzalez-2010a,gonzalez-2010b} a large collection of evaluations of def\/inite integrals that illustrate the power and f\/lexibility of this method.

For $a \in \mathbb{C}$, the symbol
\begin{gather*}
\langle a \rangle =\int_{0}^{\infty }x^{a-1} dx
\end{gather*}
is the {\em bracket} associated to the (divergent) integral on the right. The symbol
\begin{gather*}
\phi_{n} = \frac{(-1)^{n}}{\Gamma(n+1)}
\end{gather*}
is called the {\em indicator} associated to the index $n$. The notation $\phi_{i_{1}i_{2}\cdots i_{r}}$, or simply $\phi_{12 \cdots r}$, denotes the product
$\phi_{i_{1}} \phi_{i_{2}} \cdots \phi_{i_{r}}$.

{\bf {\em Rules for the production of bracket series.}}

{\bf Rule} $\boldsymbol{{\rm P}_{1}}$.
Assign to the integral $\int_{0}^{\infty }f(x)\;dx$ a bracket series:
\begin{gather*}
\sum\limits_{n}\phi _{n} a(n) \langle \alpha n+\beta \rangle .
\end{gather*}
Here the coef\/f\/icients $a(n) $ come from an assumed expansion $f(x) =\sum\limits_{n\geq 0}\phi _{n} a(n) x^{\alpha n+\beta -1}$. The extra `$-1$' in the exponent is set for convenience. The coef\/f\/icients are written as $a(n) $ because these will soon be evaluated at complex numbers $n$, not necessarily positive integers. Now we need to state how to convert the bracket series into a number.

{\bf Rule} $\boldsymbol{{\rm P}_{2}}$. For $\alpha \in \mathbb{C}$, the multinomial power $(a_{1} + a_{2} + \cdots + a_{r})^{\alpha}$ is assigned the
$r$-dimensional bracket series
\begin{gather*}
\sum_{n_{1}}\sum_{n_{2}}\cdots \sum_{n_{r}}\phi _{n_{1} n_{2} \cdots n_{r}}a_{1}^{n_{1}}\cdots a_{r}^{n_{r}}\frac{\langle -\alpha +n_{1}+\cdots +n_{r}\rangle }{\Gamma (-\alpha )}.
\end{gather*}

{\bf {\em Rules for the evaluation of a bracket series.}}

{\bf Rule} $\boldsymbol{{\rm E}_{1}}$. The one-dimensional bracket series is assigned the value
\begin{gather*}
\sum_{n} \phi_{n} f(n) \langle an + b \rangle = \frac{1}{|a|} f(n^{*}) \Gamma(-n^{*}),
\end{gather*}
where $n^{*}$ is obtained from the vanishing of the bracket; that is, $n^{*}$ solves $an+b = 0$. This is precisely the Ramanujan's master theorem.

The next rule provides a value for multi-dimensional bracket series of index $0$, that is, the number of sums is equal to the number of brackets.

{\bf Rule} $\boldsymbol{{\rm E}_{2}}$. Assume the matrix $A = (a_{ij})$ is non-singular, then the assignment is
\begin{gather*}
\sum_{n_{1}} \cdots \sum_{n_{r}} \phi_{n_{1} \cdots n_{r}} f(n_{1},\dots,n_{r})
\langle a_{11}n_{1} + \cdots + a_{1r}n_{r} + c_{1} \rangle \cdots \langle a_{r1}n_{1} + \cdots + a_{rr}n_{r} + c_{r} \rangle\\
\qquad {} = \frac{1}{| \text{det}(A) |} f(n_{1}^{*}, \dots, n_{r}^{*}) \Gamma(-n_{1}^{*}) \cdots \Gamma(-n_{r}^{*}),
\end{gather*}
where $\{ n_{i}^{*} \}$ is the (unique) solution of the linear system obtained from the vanishing of the brackets. There is no assignment if~$A$ is singular.

{\bf Rule} $\boldsymbol{{\rm E}_{3}}$. Each representation of an integral by a bracket series has associated an {\em index of the representation} via
\begin{gather*}
\text{index} = \text{number of sums} - \text{number of brackets}.
\end{gather*}
It is important to observe that the index is attached to a specif\/ic representation of the integral and not just to integral itself. The experience obtained by the authors using this method suggests that, among all representations of an integral as a bracket series, the one with {\em minimal index} should be chosen.

The value of a multi-dimensional bracket series of positive index is obtained by computing all the contributions of maximal rank by Rule $E_{2}$. These contributions to the integral appear as series in the free parameters. Series converging in a~common region are added and divergent series are discarded. Any series producing a non-real contribution is also discarded. There is no assignment to a bracket series of negative index.

\section{The evaluation of the expectations. A f\/irst bracket calculation}\label{sec-expectations}

This section describes the evaluation of the integral
\begin{gather}\label{integral-1}
I_{n,\ell,k}(\mu):= \int_{0}^{\infty} r^{2 + 2 \ell +k}e^{-2 \mu r} \big[ L_{n- \ell -1}^{2 \ell +1}(2 \mu r) \big]^{2} dr,
\end{gather}
that appeared in \eqref{mean-1} by the method of brackets. The expectation value of a power of the radial coordinate is then given by
\begin{gather*}
\big\langle r^{k} \big \rangle_{n \ell} = ( 2 \mu)^{2 \ell} A_{n \ell}^{2} I_{n,\ell,k}(\mu).
\end{gather*}

This integral can be scaled to
\begin{gather}\label{int-I}
I_{n,\ell,k}(\mu) = \frac{1}{(2 \mu)^{3 + 2 \ell + k}}\int_{0}^{\infty} t^{2 + 2 \ell +k} e^{-t} \big[ L_{n- \ell -1}^{2 \ell +1}(t) \big]^{2} dt.
\end{gather}
This does not appear in the table~\cite{gradshteyn-2015a}. The closest entry is~7.414.10:
\begin{gather*}
\int_{0}^{\infty} e^{-bx} x^{2a} \big[ L_{n}^{a}(x) \big]^{2} dx = \frac{2^{2a} \Gamma \big( a + \tfrac{1}{2} \big)
\Gamma \big( n + \tfrac{1}{2} \big) } {\pi (n!)^{2}b^{2a+1}} \Gamma(a+n+1) \, \pFq21{-n,\,a + \tfrac{1}{2}}{\tfrac{1}{2}-n}{\left( 1 - \tfrac{2}{b}\right)^{2}}.
\end{gather*}

\begin{Note}
In the evaluation of \eqref{integral-1}, it is convenient to write it as
\begin{gather*}
I_{n,\ell,k:A,B,C}(\mu):= \int_{0}^{\infty} r^{2 + 2 \ell +k} e^{-Ar} L_{n- \ell -1}^{2 \ell +1}(Br)L_{n- \ell -1}^{2 \ell +1}(Cr) dr
\end{gather*}
and then consider the limiting value as $A$, $B$, $C$ tend to $2 \mu$.
\end{Note}

The computation of \eqref{int-I} described in this section is obtained without any further identities for the Laguerre function. Next section describes the computation of the function $G_{\ell,k,s}(\mu)$, def\/ined in~\eqref{def-G}.

The f\/irst step is to compute a series representation for the factors in the integrand.

\begin{Lemma}\label{lemma-brackets}
The functions in the integrand of \eqref{integral-1} have series given by
\begin{gather*}
e^{-ax} = \sum_{n_{1}} \phi_{n_{1}} a^{n_{1}} x^{n_{1}}
\end{gather*}
and
\begin{gather*}
L_{m}^{\alpha}(x) = \Gamma(\alpha+1+m) \sum_{n_{2}} \phi_{n_{2}}\frac{x^{n_{2}}}{\Gamma(1+m-n_{2}) \Gamma(1 + \alpha + n_{2})}.
\end{gather*}
\end{Lemma}
\begin{proof}
The series of the exponential function is elementary. Indeed,
\begin{gather*}
e^{-ax} = \sum_{n_{1} \geq 0}
\frac{(-a)^{n_{1}}}{n_{1}!}x^{n_{1}}
= \sum_{n_{1} \geq 0} \frac{(-1)^{n_{1}}}{n_{1}!} (ax)^{n_{1}}
 = \sum_{n_{1}} \phi_{n_{1}} (ax)^{n_{1}}.
\end{gather*}
To evaluate the series of the Laguerre function, treat $m$ as a~real non-integer parameter, and observe that
\begin{align*}
L_{m}^{\alpha}(x) & = \frac{\Gamma(\alpha+1+m)}{\Gamma(\alpha+1) \Gamma(m+1)}
\sum_{n_{2} = 0}^{\infty} \frac{(-m)_{n_{2}}}{(\alpha+1)_{n_{2}}}
\frac{x^{n_{2}}}{n_{2}!} \\
& = \frac{\Gamma(\alpha + 1 + m)}{\Gamma(m+1)}\sum_{n_{2} = 0}^{\infty} \frac{\Gamma(n_{2}-m)}{\Gamma(-m) \Gamma(\alpha +1
+n_{2})} \frac{x^{n_{2}}}{n_{2}!}.
\end{align*}
The series for the Laguerre function now follows from the identity
\begin{gather*}
\frac{\Gamma(n_{2}-m)}{\Gamma(-m)} = (-1)^{n_{2}} \frac{\Gamma(1+m)}
{\Gamma(1+m-n_{2})}
\end{gather*}
valid for $n_{2} \in \mathbb{N}$ and $m \not \in \mathbb{N}$.
\end{proof}

The series given in Lemma \ref{lemma-brackets} are now used directly to evaluate the integral \eqref{integral-1}. This gives
\begin{gather*}
I_{n,\ell,k;A,B,C}(\mu) = \int_{0}^{\infty} r^{2 + 2 \ell + k}
\left[ \sum_{n_{1}} A^{n_{1}} \phi_{n_{1}} r^{n_{1}} \right] \\
\hphantom{I_{n,\ell,k;A,B,C}(\mu) =}{} \times \left[ \sum_{n_{2}} \frac{\Gamma(\ell + n + 1)}{\Gamma(n - \ell -n_{2})
\Gamma(2 \ell + 2 + n_{2})} \phi_{n_{2}} B^{n_{2}} r^{n_{2}} \right] \\
\hphantom{I_{n,\ell,k;A,B,C}(\mu) =}{} \times \left[ \sum_{n_{3}} \frac{\Gamma(\ell + n + 1)}{\Gamma(n - \ell -n_{3})
\Gamma(2 \ell + 2 + n_{3})} \phi_{n_{3}} C^{n_{3}} r^{n_{3}} \right] dr \\
\hphantom{I_{n,\ell,k;A,B,C}(\mu)}{} = \sum_{n_{1},n_{2},n_{3}}\int_{0}^{\infty} r^{2 + 2 \ell + k + n_{1}+n_{2}+n_{3}} dr
A^{n_{1}} B^{n_{2}} C^{n_{3}} \phi_{n_{1},n_{2},n_{3}} \\
\hphantom{I_{n,\ell,k;A,B,C}(\mu) =}{} \times \frac{\Gamma^{2}(\ell + n + 1)}{\Gamma(n- \ell - n_{2}) \Gamma(2\ell + 2 +n_{2})
\Gamma(n- \ell - n_{3}) \Gamma(2\ell + 2 +n_{3}) } \\
\hphantom{I_{n,\ell,k;A,B,C}(\mu)}{} = \sum_{n_{1},n_{2},n_{3}} \langle n_{1} + n_{2} + n_{3} + 3 + 2 \ell + k \rangle
A^{n_{1}} B^{n_{2}} C^{n_{3}} \phi_{n_{1},n_{2},n_{3}} \\
\hphantom{I_{n,\ell,k;A,B,C}(\mu) =}{} \times \frac{\Gamma^{2}(\ell + n + 1)} {\Gamma(n- \ell - n_{2}) \Gamma(2\ell + 2 +n_{2})
\Gamma(n- \ell - n_{3}) \Gamma(2\ell + 2 +n_{3}) }.
\end{gather*}

This intermediate result is stated next.

\begin{Proposition}
A bracket series for the integral $I_{n,\ell,k;A,B,C}(\mu)$ is given by
\begin{gather*}
I_{n,\ell,k;A,B,C}(\mu) = \sum_{n_{1},n_{2},n_{3}}\langle n_{1} + n_{2} + n_{3} + 3 + 2 \ell + k \rangle
A^{n_{1}} B^{n_{2}} C^{n_{3}} \phi_{n_{1},n_{2},n_{3}} \\
\hphantom{I_{n,\ell,k;A,B,C}(\mu) =}{} \times \frac{\Gamma^{2}(\ell + n + 1)}
{\Gamma(n- \ell - n_{2}) \Gamma(2\ell + 2 +n_{2})\Gamma(n- \ell - n_{3}) \Gamma(2\ell + 2 +n_{3}) }.
\end{gather*}
\end{Proposition}

The bracket series above contains one bracket and three indices, thus it is expected that the method will produce a double series as an expression for $I_{n,\ell,k;A,B,C}(\mu)$. The vanishing of the bracket gives
\begin{gather}\label{rel-index}
n_{1}+n_{2}+n_{3} = -3 - 2 \ell -k,
\end{gather}
showing the two free indices.

\textit{Solving for $n_{3}$}. Replacing $n_{3} = -n_{1}-n_{2}-t$, with $t = 2 \ell + k + 3$, in the bracket series yields the expression
\begin{gather*}
I_{n,\ell,k;A,B,C}(\mu) = \frac{\Gamma^{2}(n + \ell + 1)}{C^{t}} \\
\qquad{}\times\sum_{n_{1},n_{2} = 0}^{\infty}
\frac{\Gamma(n_{1}+n_{2}+t) \left( - \frac{A}{C} \right)^{n_{1}}\left( - \frac{B}{C} \right)^{n_{2}} }
{ \Gamma(n- \ell - n_{2}) \Gamma(2 \ell + 2 +n_{2})\Gamma(n_{1}+n_{2}+s) \Gamma(-1-k-n_{1}-n_{2}) n_{1}! n_{2}!}
\end{gather*}
with $s= n+ \ell + 3 + k$. Using \eqref{gamma-1} yields
\begin{gather*}
I_{n,\ell,k;A,B,C}(\mu) =\frac{\Gamma^{2}(n + \ell + 1) \Gamma(t)}{C^{t}\Gamma(n - \ell) \Gamma(2 \ell + 2) \Gamma(s) \Gamma(-1-k)} \\
\hphantom{I_{n,\ell,k;A,B,C}(\mu) =}{}\times\sum_{n_{1},n_{2}=0}^{\infty}
\frac{(t)_{n_{1}+n_{2}} (1- n + \ell)_{n_{2}} (k+2)_{n_{1}+n_{2}} }{(2 \ell + 2)_{n_{2}} (s)_{n_{1}+n_{2}} n_{1}!n_{2}!}
(-1)^{n_{2}}\left( \frac{A}{C} \right)^{n_{1}}\left( \frac{B}{C} \right)^{n_{2}}.
\end{gather*}
Then use
\begin{gather*}
(b)_{n_{1}+n_{2}} = (b)_{n_{2}} (b+n_{2})_{n_{1}}
\end{gather*}
to produce
\begin{gather*}
I_{n,\ell,k;A,B,C}(\mu) =\frac{\Gamma^{2}(n + \ell + 1) \Gamma(t)}{C^{t}\Gamma(n - \ell) \Gamma(2 \ell + 2) \Gamma(s) \Gamma(-1-k)} \\
\qquad{}\times\!\sum_{n_{1},n_{2}=0}^{\infty}\!\frac{(t)_{n_{2}} (t+n_{2})_{n_{1}} (1- n + \ell)_{n_{2}}
(k+2)_{n_{2}} (k+2+n_{2})_{n_{1}} }{(2 \ell + 2)_{n_{2}} (s)_{n_{2}} (s+n_{2})_{n_{1}} n_{1}!n_{2}!}(-1)^{n_{2}}
\left( \frac{A}{C} \right)^{n_{1}}\left( \frac{B}{C} \right)^{n_{2}}.
\end{gather*}
The sum corresponding to the index $n_{1}$, which appears only in $3$ places, is chosen as the internal sum. This yields
\begin{gather*}
I_{n,\ell,k;A,B,C}(\mu) =\frac{\Gamma^{2}(n + \ell + 1) \Gamma(t)}{C^{t}
\Gamma(n - \ell) \Gamma(2 \ell + 2) \Gamma(s) \Gamma(-1-k)} \\
\qquad{}\times \sum_{n_{2}=0}^{\infty}\frac{(t)_{n_{2}} (k+2)_{n_{2}} (1-n+ \ell)_{n_{2}} }
{(2 \ell + 2)_{n_{2}} (s)_{n_{2}} n_{2}! }\left( -\frac{B}{C} \right)^{n_{2}}
\sum_{n_{1}=0}^{\infty}\frac{(t+n_{2})_{n_{1}} (k+2+n_{2})_{n_{1}} }{(s+n_{2})_{n_{1}} n_{1}!}\left( \frac{A}{C} \right)^{n_{1}}.
\end{gather*}
The inner sum is now identif\/ied as a hypergeometric function to produce
\begin{gather*}
I_{n,\ell,k;A,B,C}(\mu) =\frac{\Gamma^{2}(n + \ell + 1) \Gamma(t)}{C^{t}
\Gamma(n - \ell) \Gamma(2 \ell + 2) \Gamma(s) \Gamma(-1-k)} \\
\qquad{}\times\sum_{n_{2}=0}^{\infty} \frac{(t)_{n_{2}} , (k+2)_{n_{2}} (1-n+ \ell)_{n_{2}} }
{(2 \ell + 2)_{n_{2}} (s)_{n_{2}} n_{2}! }\left( -\frac{B}{C} \right)^{n_{2}}\,
\pFq21{{t+n_{2}}, \, {2+k+n_{2}}}{s+n_{2}}{\frac{A}{C}}.
\end{gather*}

\begin{Note}
The same procedure can be used to treat the cases obtained by solving for $n_{1}$ or $n_{2}$ in the equation \eqref{rel-index}. The corresponding integrals are
\begin{gather*}
I_{n,\ell,k;A,B,C}^{(1)}(\mu) = \frac{\Gamma^{2}(n + \ell + 1) \Gamma(t)}{A^{t}\Gamma^{2}(n - \ell) \Gamma^{2}(2 \ell + 2)} \\
\hphantom{I_{n,\ell,k;A,B,C}^{(1)}(\mu) =}{}\times \sum_{n_{2}=0}^{\infty} \frac{(t)_{n_{2}} (1-n+ \ell)_{n_{2}} }
{(2 \ell + 2)_{n_{2}} n_{2}! }\left( \frac{B}{A} \right)^{n_{2}}\,
\pFq21{{t+n_{2}}, \, {1-n+\ell}}{2 \ell + 2}{\frac{C}{A}}
\end{gather*}
and
\begin{gather*}
I_{n,\ell,k;A,B,C}^{(2)}(\mu) = \frac{\Gamma^{2}(n + \ell + 1) \Gamma(t)}{A^{t}\Gamma^{2}(n - \ell) \Gamma^{2}(2 \ell + 2)} \\
\hphantom{I_{n,\ell,k;A,B,C}^{(2)}(\mu) =}{}\times\sum_{n_{3}=0}^{\infty}
\frac{(t)_{n_{3},} (1-n+ \ell)_{n_{3}} } {(2 \ell + 2)_{n_{3}} n_{3}! }
\left( \frac{C}{A} \right)^{n_{3}}\,\pFq21{{t+n_{3}}, \, {1-n+\ell}}{2 \ell + 2}{\frac{B}{A}}.
\end{gather*}
\end{Note}

At this point, the parameters $A$, $B$, $C$ are replaced by the value $2 \mu$, in order to continue the evaluation. This gives
\begin{gather*}
I_{n,\ell,k}(\mu) =\frac{\Gamma^{2}(n + \ell + 1) \Gamma(t)}{(2 \mu)^{t}\Gamma(n - \ell) \Gamma(2 \ell + 2) \Gamma(s) \Gamma(-1-k)} \\
\hphantom{I_{n,\ell,k}(\mu) =}{}\times \sum_{n_{2}=0}^{\infty}
\frac{(t)_{n_{2}} (k+2)_{n_{2}} (1-n+ \ell)_{n_{2}} }{(2 \ell + 2)_{n_{2}} (s)_{n_{2}} n_{2}! }(-1)^{n_{2}}\,
\pFq21{{t+n_{2}}, \, {2+k+n_{2}}}{s+n_{2}}{\,1}.
\end{gather*}
Observe that $1- n + \ell$ is a negative integer, so this is actually a f\/inite sum. Using Gauss' evaluation
\begin{gather*}
\pFq21{{a},{b}}{c}{1}= \frac{\Gamma(c)\Gamma(c-a-b)}{\Gamma(c-a)\Gamma(c-b)} \qquad \text{for} \quad {c-a-b} > 0,
\end{gather*}
and expressing the resulting gamma factors in terms of Pochhammer symbols to obtain
\begin{gather*}
I_{n,\ell,k}(\mu) =\frac{\Gamma(n+ \ell+1) \Gamma(2 \ell + k + 3) \Gamma(n - \ell -k -2)}{(2 \mu)^{2 \ell + k + 3}
\Gamma^{2}(n- \ell) \Gamma(2 \ell + 2) \Gamma(-1-k)} \\
\hphantom{I_{n,\ell,k}(\mu) =}{}\times \sum_{n_{2}=0}^{\infty}\frac{(k+2)_{n_{2}} (1- n + \ell)_{n_{2}} (2 \ell+ k + 3)_{n_{2}}}
{(2 \ell + 2)_{n_{2}} (\ell + k+3-n)_{n_{2}} n_{2}!}.
\end{gather*}

The f\/inal step identif\/ies this series as a hypergeometric series to produce:
\begin{gather*}
I_{n,\ell,k}(\mu) =\frac{\Gamma(n+ \ell+1) \Gamma(2 \ell + k + 3) \Gamma(n - \ell -k -2)}{(2 \mu)^{2 \ell + k + 3}
\Gamma^{2}(n- \ell) \Gamma(2 \ell + 2) \Gamma(-1-k)} \\
\hphantom{I_{n,\ell,k}(\mu) =}{}\times \pFq32{{k+2}, \,{1+\ell-n}, \,{2 \ell + k + 3}}{{2 \ell + 2}, \, {l+k+3 -n}}{\,1}.
\end{gather*}

The results of this section are summarized in the next statement.

\begin{Theorem}\label{f3}
For $n$, $\ell$, $k$ as above,
\begin{gather*}
\big\langle r^{k} \big\rangle_{n \ell} =\frac{\Gamma(2 \ell + k + 3) \Gamma(n - \ell -k - 2)}
{2n (2\mu)^{k} \Gamma(n - \ell) \Gamma(2 \ell + 2) \Gamma(-1-k)} \,
\pFq32{{k+2}, \,{1+\ell-n}, \, {2 \ell + k + 3}}{{2 \ell + 2}, \, {l+k+3 -n}}{\,1}.
\end{gather*}
\end{Theorem}

\section{The evaluation of the expectations. A second approach}\label{sec-expectations-part2}

The moment $\langle r^{k} \rangle_{n \ell}$ has been expressed in \eqref{rwithG} as a f\/inite sum values of the integral
\begin{gather}\label{g-def-00}
G_{\ell,k,s}(\mu) = \int_{0}^{\infty} r^{2 + 2 \ell + k} e^{-2 \mu r} L_{2s}^{2( 2 \ell + 1)}( 4 \mu r ) dr,
\end{gather}
where the index $s$ is an integer varying from $0$ to $n - \ell - 1$. Corollary \ref{moments-2} provides an expression for~$\langle r^{k} \rangle_{n \ell}$ as a f\/inite sum of values of the hypergeometric function ${_{2}F_{1}}$ evaluated at the argument~$2$. The hypergeometric terms appearing in the mentioned representation are actually f\/inite sums, so the convergence of the series is not an issue. An alternative form is derived in this section that extends the range of validity of $G_{\ell,k,s}(\mu)$ to a larger range for the parameter~$s$.

The goal is to produce a representation of the series for the Laguerre polynomials, given initially by
\begin{gather*}
L_{n}^{\alpha}(x) = \frac{\Gamma(\alpha + n + 1)}{\Gamma(n+1) \Gamma(\alpha + 1)}\, \pFq11{-n}{\alpha+1}{\,x}.
\end{gather*}
This series is now written in a form suitable for the application of the method of brackets:
\begin{align*}
L_{n}^{\alpha}(x) & = \frac{\Gamma(\alpha + n + 1)}{\Gamma(n+1) \Gamma(\alpha + 1)}\sum_{k_{1}=0}^{\infty} \frac{(-n)_{k_{1}}}{(\alpha+1)_{k_{1}}} \frac{x^{k_{1}}}{k_{1}!} \\
& = \frac{\Gamma(\alpha + n + 1)}{\Gamma(n+1) \Gamma(\alpha + 1)}\sum_{k_{1}=0}^{\infty} (-1)^{k_{1}} (-n)_{k_{1}} (-\alpha )_{-k_{1}} \frac{x^{k_{1}}}{k_{1}!} \\
& = \frac{\Gamma(\alpha + n + 1)}{\Gamma(n+1) \Gamma(\alpha + 1)}\sum_{k_{1}} \phi_{1} (-n)_{k_{1}} (- \alpha)_{k_{1}} x^{k_{1}} \\
& = \frac{\Gamma(\alpha + n + 1)}{\Gamma(n+1) \Gamma(\alpha+1) \Gamma(-n) \Gamma(- \alpha)}\sum_{k_{1}} \phi_{1} \Gamma( - n + k_{1}) \Gamma(- \alpha -k_{1}) x^{k_{1}}.
\end{align*}
To produce a bracket series representation of the last expression, observe that
\begin{gather*}
\Gamma(\beta) = \sum_{\ell} \phi_{\ell} \langle \beta + \ell \rangle
\end{gather*}
and this leads to
\begin{gather*}
L_{n}^{\alpha}(x) = \frac{\Gamma(\alpha + n + 1)}{\Gamma(n+1) \Gamma(\alpha+1) \Gamma(-n) \Gamma(-\alpha)}
\sum_{k_{1},k_{2}k_{3}} \phi_{123} \langle -n + k_{1} + k_{2} \rangle \langle -\alpha -k_{1} + k_{3} \rangle x^{k_{1}}.
\end{gather*}
The vanishing of the brackets provides two representations for the Laguerre function, denoted by $T_{j}$.

\textit{Case 1.} Take $k_{1}$ as a free index. Then $k_{2}^{*} = n - k_{1}$ and $k_{3}^{*} = k_{1} + \alpha$ yields the
expression{\samepage
\begin{gather*}
T_{1} = \frac{\Gamma(\alpha+n+1)}{\Gamma(n+1) \Gamma(\alpha + 1)}
\sum_{k_{1}=0}^{\infty} \frac{(-n)_{k_{1}}}{(\alpha + 1)_{k_{1}}} \frac{x^{k_{1}}}{k_{1}!}.
\end{gather*}
This is the original series for $L_{n}^{\alpha}(x)$.}

\textit{Case 2.} Take $k_{2}$ as a free index. Then $k_{1}^{*}= n - k_{2}$ and $k_{3}^{*} = \alpha + n - k_{2}$ yields
\begin{gather}\label{L-case2}
T_{2} = \frac{\Gamma(\alpha + n+1) x^{n}}{\Gamma(n+1) \Gamma(\alpha+1) \Gamma(-n) \Gamma(-\alpha)}
\sum_{k_{2}=0}^{\infty} \Gamma(-n+k_{2}) \Gamma(-\alpha -n + k_{2}) \frac{(-x)^{-k_{2}}}{k_{2}!}.
\end{gather}

\textit{Case 3.} Taking $k_{3}$ as a free index does not produce a representation for $L_n^{\alpha}(x)$.

The next step is to use the $T_{2}$ representation to evaluate the integral $G_{\ell,k,s}(\mu)$.
By equa\-tion~(\ref{L-case2}), the expression for $L_{n}^{\alpha}(x)$ is now written as
\begin{gather*}
L_{n}^{\alpha}(x) = \frac{\Gamma(\alpha + n+1) x^{n}}{\Gamma(n+1) \Gamma(\alpha+1) \Gamma(-n) \Gamma(-\alpha)}
\sum_{j=0}^{\infty} \phi_{j} \Gamma(-n+j) \Gamma(-\alpha -n + j) x^{-j}.
\end{gather*}

Using this representation in \eqref{g-def-00} produces
\begin{gather*}
G_{\ell,k,s}(\mu) =\frac{\Gamma(4\ell+3+2s)(4\mu )^{2s}}{\Gamma(2s+1)\Gamma(4\ell+3)\Gamma(-2s)\Gamma(-4\ell-2)}\\
\hphantom{G_{\ell,k,s}(\mu) =}{}\times \sum_{j=0}^{\infty}\phi_j\Gamma(-2s+j)\Gamma(-4\ell-2-2s+j)(4\mu)^{-j}\int_0^{\infty}r^{2+2\ell+k+2s-j}e^{-2\mu r} dr.
\end{gather*}
Evaluating the last integral in terms of the gamma function and simplifying produces a proof of the next result.

\begin{Theorem}
The integral
\begin{gather*}
G_{\ell,k,s}(\mu) = \int_{0}^{\infty} r^{2 + 2 \ell + k} e^{-2 \mu r} L_{2s}^{2( 2 \ell + 1)}( 4 \mu r ) dr
\end{gather*}
is given by
\begin{gather*}
G_{\ell, k, s }(\mu) = \frac{4^{s}}{(2 \mu)^{3 + 2 \ell +k}} \frac{\Gamma( 3 + 2 \ell + k + 2s)}{\Gamma(2s+1)}\,
\pFq21{-2s, \, -2s-4 \ell -2}{ -2 -2 \ell -k - 2s}{\, \frac{1}{2}}.
\end{gather*}
\end{Theorem}

\section{A couple of examples}\label{sec-examples}

The method of brackets has been used here to produce analytic expressions for the mean radius
\begin{gather*}
\big \langle r^{k} \big \rangle_{n \ell} =( 2 \mu)^{2 \ell} A_{n \ell}^{2} \int_{0}^{\infty} r^{2 + 2 \ell +k}
e^{-2 \mu r} \big[ L_{n- \ell -1}^{2 \ell +1}(2 \mu r) \big]^{2} dr,
\end{gather*}
stated f\/irst in \eqref{mean-1}. The physically relevant parameters are
\begin{gather*}
n = 0, 1, 2, \dots, \qquad 0 \leq \ell \leq n-1,\qquad k \in \mathbb{R}.
\end{gather*}

The expressions include
\begin{gather}
\big\langle r^{k} \big\rangle_{n \ell} =\frac{\Gamma(2 \ell + k + 3) (2n+2 \ell)! }{n 2^{2n-2\ell-1} (4 \ell +2)! (2 \mu)^{k} (n+ \ell)! (n-\ell -1)!}\nonumber \\
\hphantom{\big\langle r^{k} \big\rangle_{n \ell} =}{} \times\sum_{s=0}^{n- \ell-1}
\frac{\binom{n + \ell}{s} \binom{n-\ell-1}{s}}{\binom{2n+2 \ell}{2s}}\,
\pFq21{{-2(n-\ell-1-s)},\, {2 \ell + k +3}}{4 \ell + 3}{\,2},\label{formula1a-for-r}
\end{gather}
where $\langle r^{k} \rangle_{n \ell}$ is given as a f\/inite sum of hypergeometric terms and
\begin{gather*}
\big\langle r^{k} \big\rangle_{n \ell} =\frac{\Gamma(2 \ell + k + 3) \Gamma(n - \ell -k - 2)}
{2n (2\mu)^{k} \Gamma(n - \ell) \Gamma(2 \ell + 2) \Gamma(-1-k)} \,
\pFq32{{k+2}, \,{1+\ell-n}, \, {2 \ell + k + 3}} {{2 \ell + 2}, \, {l+k+3 -n}}{\,1}
\end{gather*}
given in Theorem \ref{f3}. This section compares these expressions with the results found in the literature.

\begin{Example}
Take $\ell = n - 1$. Then the sum \eqref{formula1a-for-r} reduces to $1$ since the index $s$ must vanish. Then
\begin{gather*}
\big\langle r^{k} \big\rangle_{n,n-1} =\frac{\Gamma(k+2n+1)}{(2 \mu)^{k} (2n)!}.
\end{gather*}
In particular, for $k \in \mathbb{N}$, this becomes
\begin{gather*}
\big\langle r^{k} \big\rangle_{n,n-1} =\frac{(2n+k)!}{(2 \mu)^{k}(2n)!}.
\end{gather*}
\end{Example}

\begin{Example}
The case $\ell = n-2$ reduces the sum \eqref{formula1a-for-r} to two terms. The result is
\begin{gather*}
\langle r^{k} \rangle_{n,n-2} = \frac{(k^{2}+3k+2n)\Gamma(k+2n-1)}{2 (2 \mu)^{k} (2n-2)!}.
\end{gather*}
\end{Example}

\subsection*{Acknowledgments} The work of I.K.\ was supported in part by Fondecyt (Chile) Grants Nos.\ 1040368, 1050512 and 1121030, by DIUBB (Chile) Grant Nos. 102609, GI 153209/C and GI 152606/VC. V.H.M.\ acknowledges the partial support of NSF-DMS 1112656.

\pdfbookmark[1]{References}{ref}
\LastPageEnding


\begin{thebibliography}{99}
\footnotesize\itemsep=0pt

\bibitem{Allendes:2009bd}
Allendes P., Guerrero N., Kondrashuk I., Notte-Cuello E.A., New
 four-dimensional integrals by {M}ellin--{B}arnes transform, \href{https://doi.org/10.1063/1.3357105}{\textit{J.~Math.
 Phys.}} \textbf{51} (2010), 052304, 18~pages, \href{http://arxiv.org/abs/0910.4805}{arXiv:0910.4805}.

\bibitem{Allendes:2012mr}
Allendes P., Kniehl B.A., Kondrashuk I., Notte-Cuello E.A., Rojas-Medar M.,
 Solution to {B}ethe--{S}alpeter equation via {M}ellin--{B}arnes transform,
 \href{https://doi.org/10.1016/j.nuclphysb.2013.01.012}{\textit{Nuclear Phys.~B}} \textbf{870} (2013), 243--277, \href{http://arxiv.org/abs/1205.6257}{arXiv:1205.6257}.

\bibitem{amdeberhan-2012b}
Amdeberhan T., Espinosa O., Gonzalez I., Harrison M., Moll V.H., Straub A.,
 Ramanujan's master theorem, \href{https://doi.org/10.1007/s11139-011-9333-y}{\textit{Ramanujan~J.}} \textbf{29} (2012),
 103--120.

\bibitem{andrews-1999a}
Andrews G.E., Askey R., Roy R., Special functions, \href{https://doi.org/10.1017/CBO9781107325937}{\textit{Encyclopedia of
 Mathematics and its Applications}}, Vol.~71, Cambridge University Press,
 Cambridge, 1999.

\bibitem{dehesa-2010a}
Dehesa J.S., L\'opez-Rosa S., Mart{\'{\i}}nez-Finkelshtein A., Y\'a\~nez R.J.,
 Information theory of {D}-dimensional hydrogenic systems: application to
 circular and {R}ydberg states, \href{https://doi.org/10.1002/qua.22244}{\textit{Int.~J. Quantum Chem.}} \textbf{110}
 (2010), 1529--1548.

\bibitem{gonzalez-2014a}
Gonzalez I., Kohl K.T., Moll V.H., Evaluation of entries in {G}radshteyn and
 {R}yzhik employing the method of brackets, \textit{Sci. Ser.~A Math. Sci.}
 \textbf{25} (2014), 65--84.

\bibitem{Gonzalez:2012gu}
Gonzalez I., Kondrashuk I., Belokurov--{U}syukina loop reduction in non-integer
 dimension, \href{https://doi.org/10.1134/S1063779613020135}{\textit{Phys. Part. Nuclei}} \textbf{44} (2013), 268--271,
 \href{http://arxiv.org/abs/1206.4763}{arXiv:1206.4763}.

\bibitem{Gonzalez:2012wk}
Gonzalez I., Kondrashuk I., Box ladders in a noninteger dimension,
 \href{https://doi.org/10.1007/s11232-013-0120-3}{\textit{Theoret. and Math. Phys.}} \textbf{177} (2013), 1515--1539,
 \href{http://arxiv.org/abs/1210.2243}{arXiv:1210.2243}.

\bibitem{gonzalez-2010a}
Gonzalez I., Moll V.H., Def\/inite integrals by the method of brackets.~{I},
 \href{https://doi.org/10.1016/j.aam.2009.11.003}{\textit{Adv. in Appl. Math.}} \textbf{45} (2010), 50--73, \href{http://arxiv.org/abs/0812.3356}{arXiv:0812.3356}.

\bibitem{gonzalez-2010b}
Gonzalez I., Moll V.H., Straub A., The method of brackets. {P}art~2: examples
 and applications, in Gems in Experimental Mathematics, \href{https://doi.org/10.1090/conm/517/10139}{\textit{Contemp.
 Math.}}, Vol.~517, Editors T.~Amdeberhan, L.~Medina, V.H.~Moll, Amer. Math.
 Soc., Providence, RI, 2010, 157--171, \href{http://arxiv.org/abs/1004.2062}{arXiv:1004.2062}.

\bibitem{gonzalez-2007a}
Gonzalez I., Schmidt I., Optimized negative dimensional integration method
 ({NDIM}) and multiloop {F}eynman diagram calculation, \href{https://doi.org/10.1016/j.nuclphysb.2007.01.031}{\textit{Nuclear
 Phys.~B}} \textbf{769} (2007), 124--173, \href{http://arxiv.org/abs/hep-th/0702218}{hep-th/0702218}.

\bibitem{gonzalez-2008a}
Gonzalez I., Schmidt I., Modular application of an integration by fractional
 expansion method to multiloop {F}eynman diagrams, \href{https://doi.org/10.1103/PhysRevD.78.086003}{\textit{Phys. Rev.~D}}
 \textbf{78} (2008), 086003, 27~pages, \href{http://arxiv.org/abs/0812.3625}{arXiv:0812.3625}.

\bibitem{gonzalez-2009a}
Gonzalez I., Schmidt I., Modular application of an integration by fractional
 expansion method to multiloop {F}eynman diagrams.~{II}, \href{https://doi.org/10.1103/PhysRevD.79.126014}{\textit{Phys. Rev.~D}}
 \textbf{79} (2009), 126014, 13~pages, \href{http://arxiv.org/abs/0812.3595}{arXiv:0812.3595}.

\bibitem{gradshteyn-2015a}
Gradshteyn I.S., Ryzhik I.M., Table of integrals, series, and products, 8th~ed.,
Editors D.~Zwillinger, V.~Moll, \href{https://doi.org/10.1016/B978-0-12-384933-5.00014-X}{Elsevier/Academic Press}, New York, 2014.

\bibitem{Kniehl:2013dma}
Kniehl B.A., Kondrashuk I., Notte-Cuello E.A., Parra-Ferrada I., Rojas-Medar
 M., Two-fold {M}ellin--{B}arnes transforms of {U}syukina--{D}avydychev
 functions, \href{https://doi.org/10.1016/j.nuclphysb.2013.08.002}{\textit{Nuclear Phys.~B}} \textbf{876} (2013), 322--333,
 \href{http://arxiv.org/abs/1304.3004}{arXiv:1304.3004}.

\bibitem{kohl-2011b}
Kohl K.T., Algorithmic methods for def\/inite integration, Ph.D.~Thesis, Tulane
 University, 2011.

\bibitem{suslov-2014a}
Koutschan C., Paule P., Suslov S.K., Relativistic {C}oulomb integrals and
 {Z}eilberger's holonomic systems approach.~{II}, in Algebraic and Algorithmic
 Aspects of Dif\/ferential and Integral Operators, \href{https://doi.org/10.1007/978-3-642-54479-8_6}{\textit{Lecture Notes in
 Comput. Sci.}}, Vol.~8372, Editors M.~Barkatou, Th.~Cluzeau, G.~Regensburger,
 M.~Rosenkranz, Springer, Heidelberg, 2014, 135--145, \href{http://arxiv.org/abs/1306.1362}{arXiv:1306.1362}.

\bibitem{nikiforov-1991a}
Nikiforov A.F., Suslov S.K., Uvarov V.B., Classical orthogonal polynomials of a
 discrete variable, \href{https://doi.org/10.1007/978-3-642-74748-9}{\textit{Springer Series in Computational Physics}}, Springer-Verlag,
 Berlin, 1991.

\bibitem{pasternak-1937a}
Pasternack S., On the mean value of $r^{s}$ for {K}eplerian systems,
 \href{https://doi.org/10.1073/pnas.23.2.91}{\textit{Proc. Natl. Acad. Sci. USA}} \textbf{23} (1937), 91--94.

\bibitem{suslov-2013a}
Paule P., Suslov S.K., Relativistic {C}oulomb integrals and {Z}eilberger's
 holonomic systems approach.~{I}, in Computer Algebra in Quantum Field Theory:
 Integration, Summation and Special Functionss, Editors C.~Schneider,
 J.~Bl\"umlein, \href{https://doi.org/10.1007/978-3-7091-1616-6}{\textit{Texts and Monographs in Symbolic Computation}}, Springer-Verlag,
 Wien, 2013, 225--241, \href{http://arxiv.org/abs/1206.2071}{arXiv:1206.2071}.

\bibitem{dehesa-1997a}
S{\'a}nchez-Ruiz J., Dehesa J.S., Expansions in series of orthogonal
 hypergeometric polynomials, \href{https://doi.org/10.1016/S0377-0427(97)00243-4}{\textit{J.~Comput. Appl. Math.}} \textbf{89}
 (1998), 155--170.

\bibitem{suslov-2009a}
Suslov S.K., Expectation values in relativistic Coulomb problems,
 \href{https://doi.org/10.1088/0953-4075/42/18/185003}{\textit{J.~Phys.~B: At. Mol. Opt. Phys.}} \textbf{42} (2009), 185003, 8~pages,
 \href{http://arxiv.org/abs/0906.3338}{arXiv:0906.3338}.

\bibitem{suslov-2010b}
Suslov S.K., Mathematical structure of relativistic {C}oulomb integrals,
 \href{https://doi.org/10.1103/PhysRevA.81.032110}{\textit{Phys. Rev.~A}} \textbf{81} (2010), 032110, 8~pages,
 \href{http://arxiv.org/abs/0911.0111}{arXiv:0911.0111}.

\bibitem{suslov-2010a}
Suslov S.K., Relativistic {K}ramers--{P}asternack recurrence relations,
 \href{https://doi.org/10.1088/0953-4075/43/7/074006}{\textit{J.~Phys.~B: At. Mol. Opt. Phys.}} \textbf{43} (2010), 074006, 7~pages,
 \href{http://arxiv.org/abs/0908.3021}{arXiv:0908.3021}.

\bibitem{suslov-2008a}
Suslov S.K., Trey B., The {H}ahn polynomials in the nonrelativistic and
 relativistic {C}oulomb problems, \href{https://doi.org/10.1063/1.2830804}{\textit{J.~Math. Phys.}} \textbf{49} (2008),
 012104, 51~pages, \href{http://arxiv.org/abs/0707.1887}{arXiv:0707.1887}.

\bibitem{dehesa-2016b}
Toranzo I.V., Dehesa J.S., R\'enyi, {S}hannon and {T}sallis entropies of
 {R}ydberg hydrogenic systems, \href{https://doi.org/10.1209/0295-5075/113/48003}{\textit{Europhys. Lett.}} \textbf{113} (2016),
 48003, 6~pages, \href{http://arxiv.org/abs/1603.09494}{arXiv:1603.09494}.

\bibitem{dehesa-2016a}
Toranzo I.V., Mart{\'{\i}}nez-Finkelshtein A., Dehesa J.S., Heisenberg-like
 uncertainty measures for {$D$}-dimensional hydrogenic systems at large~{$D$},
 \href{https://doi.org/10.1063/1.4961322}{\textit{J.~Math. Phys.}} \textbf{57} (2016), 082109, 21, \href{http://arxiv.org/abs/1609.01113}{arXiv:1609.01113}.

\bibitem{dehesa-2000a}
Van~Assche W., Y{\'a}{\~n}ez R.J., Gonz{\'a}lez-F{\'e}rez R., Dehesa J.S.,
 Functionals of {G}egenbauer polynomials and {$D$}-dimensional hydrogenic
 momentum expectation values, \href{https://doi.org/10.1063/1.1286984}{\textit{J.~Math. Phys.}} \textbf{41} (2000),
 6600--6613.

\bibitem{vanvleck-1934a}
van Vleck J.H., A new method of calculating the mean value of $1/r^{s}$ for
 {K}eplerian systems in quantum mechanics, \href{https://doi.org/10.1098/rspa.1934.0027}{\textit{Proc. Roy. Soc. London.
 Ser.~A}} \textbf{143} (1934), 679--681.

\end{thebibliography}
\end{document}